\documentclass{sig-alternate-2013}

\usepackage[pass]{geometry}
\mathcode`l="8000
\begingroup
\makeatletter
\lccode`\~=`\l
\DeclareMathSymbol{\lsb@l}{\mathalpha}{letters}{`l}
\lowercase{\gdef~{\ifnum\the\mathgroup=\m@ne \ell \else \lsb@l \fi}}%
\endgroup

\newfont{\authfntsmall}{phvr at 11pt}
\newfont{\eaddfntsmall}{phvr at 9pt}

\usepackage[utf8]{inputenc}   
\usepackage[T1]{fontenc}      
\usepackage{lmodern}
\usepackage{times}
\usepackage{mathptmx}
\usepackage[shortlabels]{enumitem}
\setlist[description]{font=\normalfont\itshape,itemsep=0ex,partopsep=0ex}

\usepackage{microtype} 
\usepackage{mathtools}  
\usepackage{booktabs}   
\usepackage{array}      
\usepackage[nocompress,nosort]{cite}

\usepackage{float}
\usepackage[noend]{algpseudocode}
\newfloat{algo}{tb}{lop}
\floatname{algo}{Algorithm}
\usepackage{caption}  
\usepackage{comment}  
\usepackage{amssymb,amsmath}
\usepackage{amsfonts}
\usepackage{xcolor}
\usepackage{url}
\usepackage[plainpages=false,pdfpagelabels,colorlinks=true,citecolor=blue,hypertexnames=false]{hyperref}
\usepackage{paralist}

\usepackage{theorem}

\newtheorem{thm}{Theorem}
\newtheorem{prop}[thm]{Proposition}
\newtheorem{lem}[thm]{Lemma}
\newtheorem{coro}[thm]{Corollary}

\newtheorem{fact}[thm]{Fact}
\newdef{example}{Example}

\newenvironment{algoenv}[3][\linewidth]{
\begin{minipage}{#1} %
\flushleft
\rule{\textwidth}{.08em}\vspace{-\baselineskip}\smallskip
\begin{description}[noitemsep]
\item[\rlap{Input}\phantom{Output}] #2
\item[Output] #3
\end{description}
\vspace{-\baselineskip}
\rule{\textwidth}{.05em}
\begin{algorithmic}}
{\end{algorithmic}
\vspace{-.5\baselineskip}
\rule{\textwidth}{.08em}
\end{minipage}}

\newcommand{\bigO}{{{O}}}
\newcommand{\softO}{\tilde{\bigO}}
\def\le{\leqslant}
\def\ge{\geqslant}

\def\AA{\mathbb{A}}
\def\KK{\mathbb{K}}
\def\NN{\mathbb{N}}
\def\QQ{\mathbb{Q}}
\def\ZZ{\mathbb{Z}}

\newcommand{\rec}{\operatorname{rec}}
\newcommand{\Newton}{\mathcal{N}}
\newcommand{\Residue}{\operatorname{Residue}}
\newcommand{\Resultant}{\operatorname{Resultant}}

\newcommand{\bideg}{\operatorname{bideg}}
\newcommand{\Diag}{\operatorname{Diag}}
\newcommand{\slope}{\operatorname{ddeg}}
\newcommand{\nsmall}{\operatorname{Nsmall}}
\newcommand{\val}{\operatorname{val}}

\def\gathen#1{{#1}}

\title{Algebraic Diagonals and Walks}

\numberofauthors{3}
\author{
  \alignauthor
  Alin Bostan\\
  \affaddr{Inria}\\
  \affaddr{France}\\
  \email{Alin.Bostan@inria.fr}
  \alignauthor
  Louis Dumont\\
  \affaddr{Inria}\\
  \affaddr{France}\\
  \email{Louis.Dumont@inria.fr}
  \alignauthor
  Bruno Salvy\\
  \affaddr{Inria, Laboratoire LIP}\\
  \affaddr{(U.~Lyon,~CNRS,~ENS~Lyon,~UCBL)}\\
  \affaddr{France}\\
  \email{Bruno.Salvy@inria.fr}
}

\begin{document}
\newfont{\mycrnotice}{ptmr8t at 7pt}
\newfont{\myconfname}{ptmri8t at 7pt}
\let\crnotice\mycrnotice%
\let\confname\myconfname%

\permission{Publication rights licensed to ACM. ACM acknowledges that this contribution was authored or co-authored by an employee, contractor or affiliate of a national government. As such, the Government retains a nonexclusive, royalty-free right to publish or reproduce this article, or to allow others to do so, for Government purposes only.}


\clubpenalty=10000 
\widowpenalty = 10000

\setlength{\belowdisplayskip}{.3\baselineskip} \setlength{\belowdisplayshortskip}{0pt}
\setlength{\abovedisplayskip}{.36\baselineskip} \setlength{\abovedisplayshortskip}{0pt}

\maketitle

\begin{abstract}
The diagonal of a multivariate power series $F$ is the univariate power series $\Diag F$ generated by the diagonal terms of $F$. Diagonals form an important class of power series; they occur frequently in number theory, theoretical physics and enumerative combinatorics. We study algorithmic questions related to diagonals in the case where $F$ is the Taylor expansion of a bivariate rational function. It is classical that in this case $\Diag F$ is an algebraic function. We propose an algorithm that computes an annihilating polynomial for $\Diag F$. Generically, it is its minimal polynomial and is obtained in time quasi-linear in its size. We show that this minimal polynomial has an exponential size with respect to the degree of the input rational function. We then address the related problem of enumerating directed lattice walks. The insight given by our study leads to a new method for expanding the generating power series of bridges, excursions and meanders. We show that their first $N$ terms can be computed in quasi-linear complexity in $N$, without first computing a very large polynomial equation.
\end{abstract}

\vspace{12mm}
\noindent
{\bf Categories and Subject Descriptors:} \\
\noindent I.1.2 [{\bf Computing Methodologies}]: Symbolic and
Algebraic Manipulations --- \emph{Algebraic Algorithms}

\vspace{12mm}
\noindent {\bf General Terms:} Algorithms, Theory.

\vspace{12mm}

\section{Introduction}
\smallskip\noindent{\bf Context.}
The \emph{diagonal} of a multivariate power series with coefficients~$a_{i_1,\dots,i_k}$ is the univariate power series with coefficients~$a_{i,\dots,i}$. Particularly interesting is the class of diagonals of \emph{rational} power series (ie, Taylor expansions of rational functions). In particular,  diagonals of \emph{bivariate} rational power series are always roots of  nonzero bivariate polynomials (ie, they are algebraic series)~\cite{Polya1921,Furstenberg1967}.
Since it is also classical that  algebraic series are D-finite (ie, satisfy linear differential equations with polynomial coefficients),  their coefficients satisfy  linear recurrences and this leads to an optimal algorithm for the computation of their first terms~\cite{ChudnovskyChudnovsky1986,ChudnovskyChudnovsky1987a,BostanChyzakLecerfSalvySchost2007}.
In this article, we determine the degrees of these polynomials, the cost of their computation and related applications.

\smallskip\noindent{\bf Previous work.} The algebraicity of bivariate diagonals is classical. 
The same is true for the converse; also the property persists for multivariate rational series in positive characteristic~\cite{Furstenberg1967,Safonov87,DenefLipshitz1987}.
The first occurrence we
are aware of in the literature is P\'olya's article~\cite{Polya1921}, which
deals with a particular class of bivariate rational functions; the proof uses elementary
complex analysis. Along the lines of P\'olya's approach, Furstenberg~\cite{Furstenberg1967} gave a (sketchy) proof of the general result, over the
field of complex numbers; the same argument has been enhanced later~\cite{HautusKlarner1971},\cite[\S6.3]{Stanley99}. 
Three more different proofs exist: a purely algebraic one that works 
over arbitrary fields of characteristic zero~\cite[Th. 6.1]{Gessel80} (see also~\cite[Th. 6.3.3]{Stanley99}), one based on non-commutative power series~\cite[Prop. 5]{Fliess74}, and a combinatorial proof~\cite[\S3.4.1]{BousquetMelou06}.
Despite the richness of the topic and the fact that most proofs are constructive in essence, we were not able to find in the literature any \emph{explicit} algorithm for computing a bivariate polynomial that cancels the diagonal of a general bivariate rational function.

Diagonals of rational functions appear  naturally in enumerative combinatorics. In particular, the enumeration of unidimensional walks has been the subject of recent activity, see~\cite{BanderierFlajolet2002} and the references therein.  The
algebraicity of generating functions attached to such walks is classical as
well, and related to that of bivariate diagonals. Beyond this structural
result, several quantitative and effective results are known. Explicit
formulas give the generating functions in terms of implicit algebraic functions
attached to the set of allowed steps in the case of
excursions~\cite[\S4]{BoPe00},\cite{Gessel80}, 
bridges and meanders~\cite{BanderierFlajolet2002}. Moreover, if $a$ and $b$ denote the upper and lower amplitudes of the allowed steps, the bound $d_{a,b}=\binom{a+b}{a}$ on the degrees of equations for excursions has been obtained by Bousquet-M\'elou, and showed to be tight for a specific family of step sets, as well as generically~\cite[\S2.1]{Bousquet2006}. From the algorithmic viewpoint, Banderier and
Flajolet gave an algorithm (called the
\emph{Platypus Algorithm}) for computing a polynomial of degree $d_{a,b}$ that
annihilates the generating function for excursions~\cite[\S2.3]{BanderierFlajolet2002}. 

\smallskip\noindent{\bf Contributions.}
We design (Section~\ref{sec:diagonals}) the first explicit algorithm for computing a polynomial equation for the diagonal of an arbitrary bivariate rational function. We analyze its complexity and the size of its output in Theorem~\ref{thm:bound diagonals}. The algorithm has two main steps. The first step is the computation of a polynomial equation for the residues of a bivariate rational function. We propose an efficient algorithm for this task, that is a polynomial-time version of Bronstein's algorithm~\cite{Bronstein92}; corresponding size and complexity bounds are given in Theorem~\ref{th:Bronstein}.
The second step is the computation of a polynomial equation for the sums of a fixed number of roots of a given polynomial. We design an additive version of the Platypus algorithm~\cite[\S2.3]{BanderierFlajolet2002} and analyze it in Theorem~\ref{thm:platypus-bound}. We show in Proposition~\ref{prop:generic} that generically, the size of the minimal polynomial for the diagonal of a rational function is exponential in the degree of the input and that our algorithm computes it in quasi-optimal complexity (Theorem~\ref{thm:bound diagonals}).

 In the application to walks, we show how to expand to high precision the generating functions of bridges, excursions and meanders.
Our main message is that pre-computing a polynomial equation for them is too costly, since that equation might have exponential size in the 
maximal amplitude~$d$ of the allowed steps. 
Our algorithms have quasi-linear complexity in the precision of the expansion, while keeping the pre-computation step in polynomial complexity in~$d$ 
(Theorem~\ref{thm:walks}).

\smallskip\noindent{\bf Structure of the paper.}
After a preliminary section on background and notation, we first discuss several special bivariate resultants of broader general interest in Section~\ref{sec:resultants}. Next, we consider 
diagonals, the size of their {minimal} polynomials and an efficient way of computing annihilating polynomials in Section~\ref{sec:diagonals}. 

\section{Background and Notation}
In this section, 
that might be skipped at first reading,
we introduce notation and technical results that will be used throughout the article.
\subsection{Notation}
In this article, $\KK$ denotes a field of characteristic~0. 
We denote by $\KK[x]_n$ the set of polynomials in $\KK[x]$ of degree less than~$n$. Similarly, $\KK(x)_n$ stands for the set of rational functions in $\KK(x)$ with numerator and denominator in $\KK[x]_n$, and $\KK[[x]]_n$ for the set of power series in $\KK[[x]]$ truncated at precision~$n$.

If~$P$ is a polynomial in $\KK[x,y]$, then its degree with respect to $x$ (resp. $y$) is denoted $\deg_xP$ (resp. $\deg_yP$), and the \emph{bidegree} of $P$ is the pair~$\bideg P=(\deg_xP,\deg_yP)$. The notation $\deg$ is used for univariate polynomials. Inequalities between bidegrees are component-wise. The set of polynomials in $\KK[x,y]$ of bidegree less than $(n,m)$ is denoted by $\KK[x,y]_{n,m}$, and similarly for more variables.

The \emph{valuation} of a polynomial~$F\in\KK[x]$ or a power series~$F\in\KK[[x]]$ is its smallest exponent with nonzero coefficient. It is denoted $\val F$, with the convention~$\val 0=\infty$.

The \emph{reciprocal} of a polynomial~$P\in\KK[x]$ is the polynomial $\rec(P)=x^{\deg P}P(1/x)$.
If $P=c(x-\alpha_1)\dotsm(x-\alpha_d)$, the notation $\Newton(P)$ stands for the generating series of the \emph{Newton sums} of $P$:
\[\Newton(P)=\sum_{n\ge 0}{(\alpha_1^n+\alpha_2^n+\dots+\alpha_d^n)x^n}.\]
A \emph{squarefree decomposition} of a nonzero polynomial $Q\in\AA[y]$, where $\AA=\KK$ or $\KK[x]$, is a factorization~$Q=Q_1^1\dotsm Q_m^m$, with $Q_i\in\AA[y]$ squarefree, the $Q_i$'s pairwise coprime and~$\deg_y(Q_m)>0$. The corresponding \emph{squarefree part} of~$Q$ is the polynomial~$Q^\star=Q_1\dotsm Q_m$. If $Q$ is squarefree then $Q = Q^\star$.

The coefficient of~$x^n$ in a power series $A\in\KK[[x]]$ is denoted $[x^n]A$. 
If $A=\sum_{i=0}^{\infty}{a_ix^i}$, then $A\bmod x^n$ denotes the polynomial $\sum_{i=0}^{n-1}{a_ix^i}$.
The exponential series $\sum_n x^n/n!$ is denoted $\exp(x)$.
The \emph{Hadamard product} of two power series~$A$ and~$B$ is the power series~$A\odot B$ such that $[x^n]A\odot B=[x^n]A\cdot[x^n]B$ for all $n$. 

If $F(x,y) = \sum_{i,j\ge 0}{f_{i,j}x^iy^j}$ is a bivariate power series in $\KK[[x,y]]$, 
the \emph{diagonal} of $F$, denoted $\Diag F$ is the univariate power series in~$\KK[[t]]$ defined by
$\Diag F(t) = \sum_{n\ge 0}{f_{n, n}t^n}.$

\subsection{Bivariate Power Series}
In several places, we need bounds on degrees of coefficients of bivariate rational series. In most cases, these power series belong to~$\KK(x)[[y]]$ and have a very constrained structure: there exists a polynomial~$Q\in\KK[x]$ and an integer~$\alpha\in\NN$ such that the power series can be written
\[c_0+c_1\frac{y}{Q}+\dotsb+c_n\frac{y^n}{Q^n}+\dotsb,\]
with $c_n\in\KK[x]$ and $\deg c_n\le n\alpha$, for all $n$. We denote by $\mathcal{E}_\alpha(Q)$ the set of such power series. Its main properties are summarized as follows.
\begin{lem} \label{growth}
	Let $Q, R\in\KK[x]$, $\alpha, \beta\in\NN$ and $f\in\KK[[y]]$.
	\begin{compactenum}
		\item[(1)] The set $\mathcal{E}_\alpha(Q)$ is a subring of~$\KK(x)[[y]]$;		
		\item[(2)] Let $S\in\mathcal{E}_\alpha(Q)$ with $S(0)=0$, then $f(S)\in\mathcal{E}_\alpha(Q)$;
		\item[(3)] The products obey
		\[\mathcal{E}_\alpha(Q)\cdot\mathcal{E}_\beta(R) \subset \mathcal{E}_{\max(\alpha+\deg R,\,\beta+\deg Q)}(QR).\]
	\end{compactenum} 
\end{lem}
\begin{proof} 
For~{\em(3)}, if $A = \sum_n{a_ny^n/Q^n}$ and $B = \sum_n{b_ny^n/R^n}$ belong respectively to~$\mathcal{E}_\alpha(Q)$ and $\mathcal{E}_\beta(R)$, then the $n$th coefficient of their product is a sum of terms of the form $a_i(x)Q^{n-i} b_{n-i}(x)R^i/(QR)^n$. Therefore, the degree of the numerator is  bounded by 
$i(\alpha+\deg R)+(n-i)(\beta+\deg Q)$, whence {\em(3)} is proved. Property~{\em(1)} is proved similarly.
In Property~{\em(2)}, the condition on~$S(0)$ makes~$f(S)$ well-defined. The result follows from~{\em(1)}.
\end{proof}
As consequences, we deduce the following two results.
\begin{coro}
\label{coro:invpol}
Let~$Q\in\KK[x,y]$ with $q(x)=Q(x,0)$ be such that~$q(0)\neq0$. Let $Q^\star$ be a squarefree part of $Q$. Then 
\[\frac{1}{Q}\in\frac{1}{q}\mathcal{E}_{\deg_x Q^\star}(Q^\star(x,0)).\]
\end{coro}
\begin{proof}Write $Q=q+R$ with $R/q\in\mathcal{E}_{\deg_xQ}(q)$. Then the result when $Q$ is squarefree ($Q=Q^\star$) follows from Part~{\em(2)} of Lemma~\ref{growth}, with~$f=1/(1+y)$. 
The general case then follows from Parts~{\em(1,3)}.
\end{proof}
\begin{prop}\label{prop:derivatives-rational-function}
Let~$P$ and~$Q$ be polynomials in~$\KK[x,y]$, with $Q(0,0)\neq0$, $\deg_y Q>0$ and~$F=P/Q$. Then for all $n\in\NN$,
\[\frac{d^nF}{dy^n}=\frac{A}{Q(Q^\star)^n},\]
with $\bideg A\le \bideg P+n(\deg_xQ^\star,\deg_yQ^\star-1)$.
\end{prop}
\begin{proof}
The Taylor expansion of $F(x,y+t)$ has for coefficients the derivatives of $F$. We consider it either in~$\KK(y)[x,t]$ or in ~$\KK(x)[y,t]$. Corollary~\ref{coro:invpol} applies directly for the degree in~$x$. 
The saving on the degree in~$y$ follows from observing that in the first part of the proof of the corollary, the decomposition~$Q(x,y+t)=Q(x,y)+R(x,y,t)$ has the property that $\deg_yR\le \deg_y Q-1$. This $-1$ is then propagated along the proof thanks to Part {\em(3)} of Lemma~\ref{growth}.
\end{proof}

\subsection{Complexity Estimates} \label{sec:complexity} 
We recall classical complexity notation and facts for later use. 
Let $\KK$ be again a field of characteristic zero.
Unless otherwise specified, we estimate the cost of our algorithms by counting
arithmetic operations in $\KK$ (denoted ``ops.'') at unit cost. 
The soft-O notation $\softO( \cdot)$ indicates that polylogarithmic factors
are omitted in the complexity estimates. We say that an algorithm has quasi-linear complexity if its complexity is $\softO(d)$, where $d$ is the maximal \emph{arithmetic size} (number of coefficients in~$\KK$ in a dense representation) of the input and of the output. In that case, the algorithm is said to be \emph{quasi-optimal}. 

\smallskip\noindent {\bf Univariate operations.} Throughout this article we will use the fact that most operations on polynomials, rational functions and power series in one variable can be performed in quasi-linear time.
Standard references for these questions are the books~\cite{GaGe03} and~\cite{BuClSh97}.
The needed results are summarized in Fact~\ref{fact:complexity} below.

\begin{fact}\label{fact:complexity}
The following operations can be performed in $\softO(n)$ ops. in $\KK$:
\begin{compactenum}
\item[(1)] addition, product and differentiation of elements in $\KK[x]_n$, $\KK(x)_n$ and $\KK[[x]]_n$; integration in $\KK[x]_n$ and $\KK[[x]]_n$;
\item[(2)] {extended gcd,}
 squarefree decomposition and resultant in $\KK[x]_n$;
\item[(3)] multipoint evaluation in $\KK[x]_n$, $\KK(x)_n$ at $O(n)$ points in $\KK$; interpolation in $\KK[x]_n$ and $\KK(x)_n$
from $n$ (resp. $2n-1$) values at pairwise distinct points in $\KK$;
\item[(4)] inverse, logarithm, exponential in  $\KK[[x]]_n$ (when defined);
\item[(5)] conversions between $P\in\KK[x]_n$ and $\Newton(P)\bmod x^n \in \KK[x]_n$.
\end{compactenum}
\end{fact}

\smallskip\noindent {\bf Multivariate operations.} 
Basic operations on polynomials, rational functions and power series in several variables are hard questions from the algorithmic point of view. For instance, no general quasi-optimal algorithm is currently known for computing resultants of bivariate polynomials, even though in several important cases such algorithms are available~\cite{BoFlSaSc06}.
Multiplication is the most basic non-trivial operation in this setting. The following result can be proved using Kronecker's substitution; it is quasi-optimal for fixed number of variables $m=O(1)$. 
\begin{fact}\label{fact:multiprod}
Polynomials in $\KK[x_1, \ldots, x_m]_{d_1, \ldots, d_m}$ 
and power series in $\KK[[x_1, \ldots, x_m]]_{d_1, \ldots, d_m}$
can be multiplied using \sloppy $\softO(2^m d_1 \cdots d_m)$ ops.
\end{fact}

A related operation is multipoint evaluation and interpolation. The simplest case is when the evaluation points form an $m$-dimensional tensor product grid $I_1 \times \cdots \times I_m$, where $I_j$ is a set of cardinal $d_j$. 

\begin{fact}\label{fact:multieval}\cite{Pan94}
Polynomials in $\KK[x_1, \ldots, x_m]_{d_1, \ldots, d_m}$ can be evaluated and interpolated from values that they take on $d_1 \cdots d_m$ points that form an $m$-dimensional tensor product grid using $\softO(m d_1 \cdots d_m)$ ops.
\end{fact}

\noindent Again, the complexity in Fact~\ref{fact:multieval} is quasi-optimal for fixed~$m=O(1)$.

A general (although  non-optimal) technique to deal with more involved operations on multivariable algebraic objects (eg, in $\KK[x,y]$)
is to use (multivariate) evaluation and interpolation on polynomials and to perform operations on the evaluated algebraic objects using Facts~\ref{fact:complexity}--\ref{fact:multieval}. To put this strategy in practice, the size of the output needs to be well controlled. We illustrate this philosophy on the example of resultant computation, based on the following easy variation of~\cite[Thm.~6.22]{GaGe03}.
\begin{fact} \label{borne resultant}
  Let $P(x,y)$ and $Q(x,y)$ be bivariate polynomials of respective bidegrees $(d_x^P, d_y^P)$ and $(d_x^Q, d_y^Q)$.
Then, $$\deg\Resultant_y(P(x,y), Q(x,y)) \le d_x^Pd_y^Q+d_x^Qd_y^P.$$
\end{fact}

\begin{lem} \label{algo resultant}
	Let $P$ and $Q$ be polynomials in $\KK[x_1, \ldots, x_m, y]_{d_1, \ldots, d_m, d}$. 
	Then $R = \Resultant_y(P, Q)$ belongs to $\KK[x_1, \ldots, x_m]_{D_1, \ldots, D_m}$, where $D_i = 1+2(d-1)(d_i-1)$. Moreover, the coefficients of $R$ can be computed using $\softO(2^m d_1 \cdots d_m d^{m+1})$ ops. in $\KK$. \end{lem}

\begin{proof}
	The degrees estimates follow from Fact~\ref{borne resultant}. 
To compute $R$, we use an evaluation-interpolation scheme:  $P$ and $Q$ are evaluated at $D=D_1 \cdots D_m$ points $(x_1, \ldots, x_m)$ forming an $m$ dimensional tensor product grid; $D$ univariate resultants in $\KK[y]_d$ are computed; $R$ is recovered by interpolation. By Fact~\ref{fact:multieval}, the evaluation and interpolation steps are performed in $\softO(m D)$ ops. The second one has cost $\softO(d D)$. Using the inequality $D \le 2^m d_1 \cdots d_m d^m $ concludes the proof.
\end{proof}.

We conclude this section by recalling a complexity result for the computation of a squarefree decomposition of a bivariate polynomial.

\begin{fact}\label{fact:sqfree}\cite{Lecerf08}
A squarefree decomposition of a polynomial in $\KK[x,y]_{d_x, d_y}$ can be computed using $\softO(d_x^2  d_y)$ ops.
\end{fact}

\section{Special Resultants}\label{sec:resultants}
\subsection{Polynomials for Residues}
We are interested in a polynomial that vanishes at the residues of a given rational function. It is a classical result in symbolic integration that in the case of simple poles, there is a resultant formula for such a polynomial, first introduced by Rothstein~\cite{Rothstein1976} and Trager~\cite{Trager1976}. This was later generalized by Bronstein~\cite{Bronstein92} to accommodate multiple poles as well. However, as mentioned by Bronstein, the complexity of his method grows exponentially with the multiplicity of the poles. Instead, we develop in this section an algorithm with polynomial complexity.

Let $f = P/Q$ be a nonzero element in $\KK(y)$, where $P, Q$ are two coprime
polynomials in $\KK[y]$. 
Let $Q_1Q_2^2\cdots Q_m^m$ be a squarefree decomposition of~$Q$.
For $i\in\{1,\dots,m\}$, if $\alpha$ is a root of~$Q_i$ in an algebraic extension of~$\KK$, then it is simple and the residue of~$f$ at~$\alpha$ is the coefficient of~$t^{-1}$ in the Laurent expansion of $f(\alpha+t)$ at~$t=0$. If~$V_i(y,t)$ is the polynomial $(Q_i(y+t)-Q_i(y))/t$,  this residue is the coefficient of~$t^{i-1}$ in the Taylor expansion at~$t=0$ of the regular rational function $f(y+t)Q_i^i(y+t)/V_i^i(y,t)$, 
computed with rational operations only and then evaluated at~$y=\alpha$. If this coefficient is denoted~$S_{i-1}(y)=A_i(y)/B_i(y)$, with polynomials $A_i$ and $B_i$,  the residue at~$\alpha$ is a root of $\Resultant_y(A_i-zB_i,Q_i)$. When $m=1$, this is exactly the Rothstein-Trager resultant.
This computation leads to Algorithm~\ref{algo:Bronstein}, which avoids the exponential blowup of the complexity that would follow from a symbolic pre-computation of the Bronstein resultants.

\begin{algo}	
	Algorithm \textbf{AlgebraicResidues}$(P/Q)$
	
	\begin{algoenv}{Two polynomials $P$ and $Q\in \KK[y]$}{A polynomial in $\KK[z]$ canceling all the residues of $P/Q$}
		\State Compute $Q_1Q_2^2\dotsm Q_m^m$ a squarefree decomposition of~$Q$;
		\For{$i\gets1$ to $m$}
		\If{$\deg_yQ_i=0$} \ $R_i\gets1$
		\Else
		\State $U_i(y)\gets Q(y)/Q_i^i(y)$;
		\State $V_i(y,t)\gets (Q_i(y+t)-Q_i(y))/t$;
		\State Expand $\frac{P(y+t)}{U_i(y+t)V_i^i(y,t)}=S_0+\dotsb+S_{i-1}t^{i-1}+O(t^{i})$;
    \State Write $S_{i-1}$ as $A_i(y)/B_i(y)$ with $A_i$ and $B_i$ coprime;
		\State $R_i(z)\gets \Resultant_y(A_i-zB_i,Q_i)$;
		\EndIf
		\EndFor
		\State\Return 			
		$R_1R_2\dotsm R_m$ 
	\end{algoenv}
	\caption{Polynomial canceling the residues}\label{algo:Bronstein}
\end{algo}

\begin{example}\label{ex:Bronstein}
Let $d\ge 0$ be an integer, and let $G_d(x,y)\in\QQ(x)[y]$ be the rational function $y^d/(y-y^2-x)^{d+1}$. 
The poles have order~$d+1$. In this example, the algorithm can be performed by hand for arbitrary~$d$:
a squarefree decomposition has~$m=d+1$ and~$Q_{m}=y-y^2-x$, the other $Q_i$'s being~1. Then~$V_m=1-2y-t$ and the next step is to expand
\[\frac{(y+t)^d}{(1-2y-t)^{d+1}}=\frac{(y+t)^d}{(1-2y)^{d+1}\left(1-\frac{t}{1-2y}\right)^{d+1}}.\]
Expanding the binomial series gives the coefficient of~$t^d$ as~$\frac{A_m}{B_m}$, with
\[A_m=\sum_{i=0}^d{\binom{d}{i}\binom{d+i}{i}y^i(1-2y)^{d-i}},\quad B_m=(1-2y)^{2d+1}.\]
The residues are then cancelled by~$\Resultant_y(A_m-zB_m,Q_{m})$, namely 
\begin{equation}\label{eq:pol-diag-ex}
(1-4t)^{2d+1} z^2 - \left( \sum_{k=0}^{\lfloor d/2 \rfloor} \binom{d}{2k}\binom{2k}{k} t^k \right)^2.
\end{equation}
\end{example}

\smallskip\noindent\textbf{Bounds.}
In our applications, as in the previous example, the polynomials $P$ and~$Q$ have coefficients that are themselves polynomials in another variable~$x$. Let then~$(d_P,e_P)$, $(d_Q,e_Q)$, $(d^\star,e^\star)$ and $(d_i,e_i)$ be the bidegrees in $(x,y)$ of $P$, $Q$, $Q^\star$ and $Q_i$, where~$Q^\star=Q_1\dotsm Q_m$ is a squarefree part of~$Q$. In Algorithm~\ref{algo:Bronstein},  $V_i$ has degree at most~$d_i$ in $x$ and total degree~$e_i-1$ in $(y,t)$. Similarly, $P(y+t)$ has degree~$d_P$ in~$x$ and total degree~$e_P$ in $(y,t)$. When $e^\star>1$, by Proposition~\ref{prop:derivatives-rational-function}, the coefficient~$S_j$ in the power series expansion of $P(y+t)/U_i(y+t)/V_i(y,t)^i$ has  denominator of bidegree bounded by~$(d_Q+jd^\star,e_Q-i+j(e^\star-1))$ and numerator of bidegree bounded by~$(d_P+jd^\star,e_P-j+j(e^\star-1))$. Thus by Fact~\ref{borne resultant}, $\deg_x R_i$ is at most
\begin{multline*}
((i-1)d^\star+\max(d_P,d_Q))e_i+\\
d_i((i-1)(e^\star-1)-i+\max(e_P+1,e_Q)),
\end{multline*}
while its degree in~$z$ is bounded by the number of residues~$e_i$.
Summing over all~$i$ leads to the bound
\begin{multline*}
(e_Q-e^\star)d^\star+(d_Q-d^\star)(e^\star-1)\\
+e^\star\max(d_P,d_Q)-d_Q+d^\star\max(e_P+1,e_Q).
\end{multline*}
If~$e^\star=1$, a direct computation  gives the bound $\max(d_P,d_Q)+d^\star e_P$.
\begin{thm}\label{th:Bronstein}
	Let $P(x,y)/Q(x,y)\in\KK(x,y)_{d_x+1,d_y+1}$. Let $Q^\star$ be a squarefree part of $Q$ wrt y. Let~$(d_x^\star,d_y^\star)$ be  bounds on the bidegree of $Q^\star$. Then the polynomial computed by Algorithm~\ref{algo:Bronstein} annihilates the residues of~$P/Q$, has degree in~$z$ bounded by~$d_y^\star$ and degree in~$x$ bounded by 
\[2d_x^\star(d_y+1)+(2d_y^\star-1)d_x-2d_x^\star d_y^\star.\]
It can be computed in $\bigO(m^2d_x^\star d_y^\star(m^2 + {d_y^\star}^2))$ operations in~$\KK$.
\end{thm}
Note that both bounds above (when $e^\star>1$ and $e^\star=1$) are upper bounded by~$2d_xd_y$, independently of the multiplicities.
The complexity is also bounded independently of the multiplicities by
$\bigO(d_x^\star d_y^\star d_y^4)$.

\begin{proof}
The bounds on the bidegree of $R = R_1 R_2 \cdots R_m$ are easily derived from the previous discussion.

By Fact~\ref{fact:sqfree}, a squarefree decomposition of $Q$ can be computed using $\softO(d_x^2  d_y)$ ops.
We now focus on the computations performed inside the $i$th iteration of the loop.
Computing $U_i$ requires an exact division of polynomials of bidegrees at most $(d_x, d_y)$; this division can be performed by evaluation-interpolation in $\softO(d_x d_y)$ ops. Similarly, the trivariate polynomial $V_i$ can be computed by evaluation-interpolation wrt $(x,y)$ in time $\softO(d_i e_i^2)$. By the discussion preceding Theorem~\ref{th:Bronstein}, both $A_i(x,y)$ and $B_i(x,y)$ have bidegrees at most $(D_i, E_i)$, where $D_i = d_x+id_x^\star$ and $E_i =  d_y+id_y^\star$. They can be computed by evaluation-interpolation in 
$\softO(i D_i E_i)$ ops. Finally, the resultant $R_i(x,z)$ has bidegree at most $(d_i E_i + e_i D_i, e_i)$, and since the degree in~$y$ of $A_i-zB_i$ and~$Q_i$ is at most~$E_i$, it can be computed by evaluation-interpolation in $\softO((d_i E_i + e_i D_i) e_i E_i)$ ops by Lemma~\ref{algo resultant}. The total cost of the loop is thus $\softO(L)$, where
\[ L = \sum_{i=1}^m \left( (i+e_i^2) D_i E_i + d_i e_i E_i^2 \right).\]
Using the (crude) 
bounds $D_i \le D_m$, $E_i \le E_m$, $\sum_{i=1}^m e_i^2 \le {d_y^\star}^2$ and $\sum_{i=1}^m d_i e_i \le d_x^\star d_y^\star$ shows that $L$ is bounded by
\[ 
D_m E_m \sum_{i=1}^m (i+e_i^2) + E_m^2 \sum_{i=1}^m  d_i e_i \le 
D_m E_m (m^2+{d_y^\star}^2) + E_m^2   d_x^\star d_y^\star,
\]
which, by using the inequalities $D_m \le 2m d_x^\star$ and $E_m \le 2m d_y^\star$, is seen to belong to $O(m^2d_x^\star d_y^\star(m^2 + {d_y^\star}^2))$.

Gathering together the various complexity bounds yields the stated bound and finishes the proof of the theorem.
\end{proof}

\noindent {\bf Remark.} Note that one could also use Hermite reduction combined with the usual Rothstein-Trager resultant in order to compute a polynomial $\tilde{R}(x,z)$ that annihilates the residues. Indeed, Hermite reduction computes an auxiliary rational function that admits the same residues as the input, while only having simple poles. A close inspection of this approach provides the same bound $d_y^\star$ for the degree in $y$ of $\tilde{R}(x,z)$, but a less tight bound for its degree in $x$, namely worse by a factor of $d_y^\star$. The complexity of this alternative approach appears to be $\softO(d_x d_y (d_y + {d_y^\star}^3))$ (using results from~\cite{BoChChLi10}), to be compared with the complexity bound from Theorem~\ref{th:Bronstein}.

\subsection{Sums of roots of a polynomial} \label{subsec:summation of residues}

Given a polynomial~$P\in\KK[y]$ of degree~$d$ with coefficients in a field~$\KK$ of characteristic~0, let 
$\alpha_1,\dots,\alpha_d$ be its roots in the algebraic closure of~$\KK$. For any positive integer $c\le d $, the polynomial of degree $\binom{d}{c}$ defined by
\begin{equation}\label{eq:defSigma}
\Sigma_c P = \prod_{i_1 < \cdots < i_c}{\left(y-(\alpha_{i_1}+\alpha_{i_2}+\cdots+\alpha_{i_c})\right)}
\end{equation}
has coefficients in~$\KK$. This section discusses the computation of~$\Sigma_cP$ summarized in Algorithm~\ref{algo:Sigma_c}, which can be seen as an additive analogue of the \emph{Platypus algorithm} of Banderier and Flajolet~\cite{BanderierFlajolet2002}.

\begin{algo}
	Algorithm \textbf{PureComposedSum}$(P, c)$
	
	\begin{algoenv}{A polynomial $P$ of degree $d$ in $\KK[y]$, a positive integer $c\le d$}{The polynomial $\Sigma_c P$ from Eq.~\eqref{eq:defSigma}}
		\State $D \gets\binom{d}{c}$
		\State $\Newton(P) \gets \rec(P')/\rec(P)\bmod y^{D+1}$
		\State $S \gets \Newton(P)\odot\exp(y) \bmod y^{D+1}$
		\State $F\gets\exp\left(\sum_{n=1}^c(-1)^{n-1}\frac{S(n y)}{n}z^n\right)\bmod (y^{D+1},z^{c+1})$
		\State $\Newton(\Sigma_c P)\gets ([z^c]F)\odot\sum{n!y^n}\bmod y^{D+1}$
		\State\Return 			
$\rec\left(\exp\left(\int\frac{D-\Newton(\Sigma_c P)}{y}\,\mathrm{d}y\right)\bmod y^{D+1}\right)$
	\end{algoenv}
	\caption{Polynomial canceling the sums of $c$ roots}\label{algo:Sigma_c}
\end{algo}

We recall two classical formulas (see, eg, \cite[\S2]{BoFlSaSc06}), the second one being valid for monic $P$ only::
\begin{equation}\label{eq:NewtonSums}
\Newton(P) = \frac{\rec(P')}{\rec(P)},\qquad
\rec(P) = \exp \left(\int{\frac{d-\Newton(P)}{y}\,\mathrm{d}y}\right).
\end{equation}
Truncating these formulas at order~$d+1$ makes~$\Newton(P)$ a representation of the polynomial~$P$ (up to normalization), since both conversions above can be performed quasi-optimally by Newton iteration~\cite{Schonhage1982, Pan2000a, BoFlSaSc06}.
The key for Algorithm~\ref{algo:Sigma_c} is the following variant of~\cite[\S2.3]{BanderierFlajolet2002}.
\begin{prop}
	Let~$P\in\KK[y]$ be a polynomial of degree~$d$, 
let $\Newton(P)$ denote the generating series of its Newton sums and let $S$ be the series $\Newton(P) \odot \exp(y)$.
Let $\Psi_c$ be the polynomial in $\KK[t_1, \ldots, t_c]$ defined by
	\[\Psi_c(t_1,..., t_c) = [z^c]\exp\left(\sum_{n \ge 1}{(-1)^{n-1}t_n\frac{z^n}{n}} \right).\]
	Then the following equality holds
	\[\Newton(\Sigma_c P)\odot \exp(y) = \Psi_c(S(y), S(2 y), \ldots, S(c y)).\]
\end{prop}
\begin{proof}
	By construction, the series $S$ is 
	\[S(y)=\sum_{n\ge 0}{(\alpha_1^n+\alpha_2^n+\cdots+\alpha_d^n)\frac{y^n}{n!}} = \sum_{i=1}^{d}\exp(\alpha_i y).\]
	When applied to the polynomial~$\Sigma_c P$, this becomes
	\begin{align*}
	\Newton(\Sigma_c P)\odot \exp(y) &= \sum_{i_1 < \cdots < i_c}{\exp \left({(\alpha_{i_1} + \alpha_{i_2}+\cdots+\alpha_{i_c})y}\right)}\\
	&= [z^c]\prod_{i=1}^{d}{\left(1+z \exp(\alpha_i y)\right)}.
	\end{align*}
This expression rewrites:
	\begin{multline*}
[z^c]\exp\left(\sum_{i=1}^{d}\log(1+z \exp({\alpha_i y}))\right)\\
	= [z^c]\exp\left(\sum_{i=1}^{d}\sum_{m\ge 1}{(-1)^{m-1}\exp({\alpha_i m y})\frac{z^m}{m}}\right)\\
	= [z^c]\exp\left(\sum_{m\ge 1}{(-1)^{m-1}S(m y)\frac{z^m}{m}}\right),
	\end{multline*}
and the last expression equals $\Psi_c(S(y), S(2 y), \dots, S(c y))$.
\end{proof}

The correctness of Algorithm~\ref{algo:Sigma_c} follows from observing that the truncation orders $D+1$ in $y$ and $c+1$ in $z$ of the power series involved in the algorithm are sufficient to enable the reconstruction of~$\Sigma_cP$ from its first Newton sums by~\eqref{eq:NewtonSums}.

\smallskip \noindent \textbf{Bivariate case.}
We now consider the case where~$P$ is a polynomial in~$\KK[x,y]$. Then, the coefficients of $\Sigma_cP$ wrt $y$ may have denominators. We follow the steps of Algorithm~\ref{algo:Sigma_c} (run on $P$ viewed as a polynomial in~$y$ with coefficients in~$\KK(x)$) in order to compute bounds on the bidegree of the polynomial obtained by clearing out these denominators.
We obtain the following result.

\begin{thm} \label{thm:platypus-bound}
	Let $P \in \KK[x,y]_{d_x+1,d_y+1}$, let $c$ be a positive integer such that $c\le d_y$ and let $D=\binom{d_y}{c}$. Let $a\in\KK[x]$ denote the leading coefficient of $P$ wrt $y$ and let $\Sigma_c P$ be defined as in~Eq.~\eqref{eq:defSigma}. Then $a^{D} \cdot \Sigma_c P$ is a polynomial in $\KK[x,y]$ of bidegree at most $\left(d_xD, D\right)$
that cancels all sums $\alpha_{i_1}+\cdots+\alpha_{i_c}$ of $c$ roots $\alpha_i(x)$ of $P$, with $i_1 < \cdots < i_c$.
Moreover, this polynomial can be computed in $\softO(c d_x D^2)$ ops.
\end{thm}

\noindent This result is close to optimal. Experiments
suggest that for generic $P$ of bidegree $(d_x, d_y)$ the minimal polynomial of
$\alpha_{i_1}+\cdots+\alpha_{i_c}$ has bidegree $\left(d_x\binom{d_y-1}{c-1},
\binom{d_y}{c} \right)$. In particular, our degree bound is precise in $y$, and overshoots by a factor of 
${d_y}/{c}$ only in $x$. Similarly, the complexity result is quasi-optimal up to a factor of $d_x d_y$ only.

\begin{proof}
	The Newton series $\mathcal{N}(P)$ has the form
	$$\mathcal{N}(P) = \frac{a\deg_yP+yA(x,y)}{a - yB(x,y)} = \frac{a\deg_yP+yA(x,y)}{a}\sum_{n\ge 0}\frac{y^nB(x,y)^n}{a^n},$$
	with $\deg_xA, \deg_xB \le d_x$. Since both factors belong to $\mathcal{E}_{d_x}(a)$, Lemma~\ref{growth} implies that $\mathcal{N}(P)\in\mathcal{E}_{d_x}(a)$. Applying this same lemma repeatedly, we get that $\Sigma_c P\in\mathcal{E}_{d_x}(a)$ (stability under the integration of Algorithm~\ref{algo:Sigma_c} is immediate).
	Since $\Sigma_c P$ has degree $D$ wrt $y$, we deduce that $a^{D}\Sigma_cP$ is a polynomial that satisfies the desired bound.
 	By evaluation and interpolation at $1 + d_x D$ points, and Newton iteration for quotients of power series in $\KK[[y]]_{1+D}$ (Fact~\ref{fact:complexity}), the power series $\mathcal{N}(P)$ can be computed in $\softO(d_x D^2)$ ops. The power series $S$ is then computed from $\mathcal{N}(P)$ in $O(d_x D^2)$ ops. To compute $F$ we use evaluation-interpolation wrt $x$ at $1 + d_x D$ points, and fast exponentials of power series (Fact~\ref{fact:complexity}).
The cost of this step is $\softO(c d_x D^2)$ ops. Then, $\Newton(\Sigma_c P)$ is computed for $O(d_x D^2)$ additional ops. The last exponential is again computed by evaluation-interpolation and Newton iteration using $\softO(d_x D^2)$ ops.	
\end{proof}

\section{Diagonals}\label{sec:diagonals}

\subsection{Algebraic equations for diagonals} \label{subsec:Algebraic equations for diagonals}
The relation between diagonals of bivariate rational functions and algebraic series is classical~\cite{Furstenberg1967,Polya1921}. We recall here the usual derivation when~$\KK=\mathbb C$ while setting our notation.

Let $F(x,y)$ be a rational function in~$\mathbb{C}(x,y)$, whose denominator does not vanish at $(0,0)$.
Then the diagonal of $F$ is a convergent power series that can be represented for small enough~$t$ by a Cauchy integral
\[\Diag F(t)=\frac1{2\pi i}\oint{F(t/y,y)\frac{\mathrm{d}y}{y}},\]
where the contour is for instance a circle of radius~$r$ inside an annulus where~$(t/y,y)$ remains in the domain of convergence of~$F$.
This is the basis of an algebraic approach to the computation of the diagonal as a sum of residues of the rational function
\[\frac{P(t,y)}{Q(t,y)}:=\frac{1}{y}F\left(\frac{t}{y},y\right),\]
with~$P$ and~$Q$ two coprime polynomials. For~$t$ small enough, the circle can be shrunk around~0 and only the roots of~$Q(t,y)$ tending to~0 when~$t\rightarrow0$ lie inside the contour~\cite{HautusKlarner1971}. These are called the \emph{small branches}. 
Thus the diagonal is given as
\begin{equation}\label{eq:diagonal-as-residues}
\Diag F(t) = \sum_{\substack{Q(t,y_i(t)) = 0\\ \lim\limits_{t\rightarrow 0}{y_i(t)}=0}}{\Residue\left(\frac{P(t,y)}{Q(t,y)},y=y_i(t)\right)},
\end{equation}
where the sum is over the \emph{distinct} roots of~$Q$ tending to~0. We call their number the number of small branches of~$Q$ and denote it by~$\nsmall(Q)$.

Since the~$y_i$'s are algebraic and finite in number and residues are obtained by series expansion, which entails only rational operations, it follows that the diagonal is algebraic too. 
Combining the algorithms of the previous section gives Algorithm~\ref{algo:diagonal} that produces a polynomial equation for $\Diag F$. The correctness of this algorithm over an arbitrary field of characteristic~0 follows from an adaptation of
the arguments of Gessel and Stanley~\cite[Th.~6.1]{Gessel80},\cite[Th.~6.3.3]{Stanley99}.

\begin{example}
Let $d\ge 0$ be an integer, and let $F_d(x,y)$ be the rational function $1/(1-x-y)^{d+1}$. The diagonal of $F_d$ is equal to
\[ \sum_{n \ge 0} \binom{2n+d}{n}\binom{n+d}{d}t^n.\]
By the previous argument, it is an algebraic series, which is the sum of the residues of the rational function~$G_d$ of Example~\ref{ex:Bronstein} over its small branches (with $x$ replaced by $t$). In this case, the denominator is~$y-t-y^2$. It has one solution tending to~0 with~$t$; the other one tends to~$1$. Thus the diagonal is cancelled by the quadratic polynomial~\eqref{eq:pol-diag-ex}.
\end{example}

\begin{algo}
	Algorithm \textbf{AlgebraicDiagonal}($A/B$)

	\begin{algoenv}{Two polynomials $A$ and $B\in\KK[x,y]$, with~$B(0,0)\neq0$}{A polynomial $\Phi\in\KK[t,\Delta]$ such that $\Phi(t, \Diag A/B)=0$}
		\State $G \gets \frac{1}{y}\frac{A}{B}(\frac{t}{y},y)$
    \State Write $G$ as $P/Q$ with coprime polynomials $P$ and $Q$;
    \State $R(z) \gets \textbf{AlgebraicResidues}(P/Q)$
    \State $c \gets$ number of small branches of $Q$
		\State $\Phi(t,z) \gets \operatorname{numer}(\textbf{PureComposedSum}(R, c))$
		\State\Return $\Phi(t,\Delta)$
	\end{algoenv}
	\caption{Polynomial canceling the diagonal of a rational function}\label{algo:diagonal}
\end{algo}

\begin{example}For an integer~$d>0$, we consider the rational function
\[F_d(x,y)=\frac{x^{d-1}}{1-x^{d}-y^{d+1}},\]
of bidegree~$(d,d+1)$. The first step of the algorithm produces
\[G_d(t,y)=\frac{t^{d-1}}{y^{d}-t^{d}-y^{2d+1}},\]
whose denominator is irreducible with $d$ small branches. Running Algorithm~\ref{algo:diagonal} on this example, we obtain a polynomial~$\Phi_d$ annihilating $\Diag F_d$, which is experimentally irreducible and whose bidegrees for $d=1,2,3,4$ are $(2,3),(18,10),(120,35),(700,126)$.
From these values, it is easy to conjecture that the bidegree is given by~\[\left(d(d+1)\binom{2d-1}{d-1},\binom{2d+1}{d}\right),\] of exponential growth in the bidegree of $F_d$. In general, these bidegrees do not grow faster than in this example. In Theorem~\ref{thm:bound diagonals}, we prove bounds that are barely larger than the values above.
\end{example}

\subsection{Degree Bounds and Complexity}
The rest of this section is devoted to the derivation of bounds on the complexity of Algorithm~\ref{algo:diagonal} and on the size of the polynomial it computes, which are given in Theorem~\ref{thm:bound diagonals}.

\smallskip \noindent \textbf{Degrees.}
A bound on the bidegree of $\Phi$ will be obtained from the bounds successively given by Theorems~\ref{th:Bronstein} and~\ref{thm:platypus-bound}.

In order to follow the impact of the change of variables in the first step, we define the \emph{diagonal degree} of a polynomial $P(x,y) =\sum_{i,j}{a_{i,j}x^iy^j}$ as the integer
$\slope(P) := \sup\left\{i-j\ |\ a_{i,j}\neq 0\right\}.$ 
We collect the properties of interest in the following.
\begin{lem}\label{lemma:slope}
For any~$P$ and~$Q$ in~$\KK[x,y]$,
\begin{compactenum}
	\item[(1)] $\slope(P)\le\deg_xP$;
	\item[(2)] $\slope(PQ)=\slope(P)+\slope(Q)$;
	\item[(3)] there exists a polynomial $\tilde{P}\in\KK[x,y]$, such that\\
	$P(x/y,y)=y^{-\slope(P)}\tilde{P}(x,y)$, with $\tilde{P}(x,0)\neq0$ and
  \[\bideg(\tilde{P})\le\bideg(P)+(0,\slope(P));\]
	\item[(4)] $\bideg((\tilde{P})^\star)=(\deg_xP^\star,\slope(P^\star)+\deg_yP^\star)$.
\end{compactenum}
\end{lem}
\begin{proof}
Part~{\em(1)} is immediate. The quantity $\slope(P)$ is nothing else than $-\val_yP(x/y,y)$, which makes Parts~{\em(2)} and~{\em(3)} clear too. From there, we get the identity $\widetilde{PQ}=\tilde{P}\tilde{Q}$ for arbitrary~$P$ and~$Q$, whence $(\tilde{P})^\star=\widetilde{P^\star}$ and Part~{\em(4)} is a consequence of Parts~{\em(1)} and~{\em(3)}.
\end{proof}

Thus, starting with a rational function $F=A/B\in\KK(x,y)$, with $(d_x, d_y)$ a bound on the bidegrees of $A$ and $B$, and $(d_x^\star,d_y^\star)$ a bound on the bidegree of a squarefree part $B^\star$ of $B$, the first step of the algorithm constructs $G(t,y)=y^{\alpha}\frac{P}{Q}$, with polynomials~$P$ and~$Q$ and
\begin{gather}
\alpha = \slope(B) -\slope(A) - 1 \label{eq:defalpha}\\
\notag\bideg P \le  (d_x, \slope(A) + d_y),\quad 
\bideg Q \le  (d_x, \slope(B) + d_y),\\
\notag \bideg Q^\star \le(d_x^\star,d_x^\star+d_y^\star).
\end{gather}
These inequalities give bounds on the degrees in~$x$ of the numerator and denominator of~$G$. 

The rest of the computation depends on the sign of~$\alpha$.
If~$\alpha\ge0$, then the degrees in~$y$ of $y^\alpha P$ and $Q$ are bounded by $\slope(B)+d_y$, while if~$\alpha<0$, those of $P$ and $y^{-\alpha}Q$ are bounded by $\slope(A)+d_y+1$. Thus in both cases they are bounded by $d_x+d_y+\epsilon$, where
\begin{equation}
\epsilon=\begin{cases}1\qquad\text{if $\alpha<0$,}\\
	0\qquad\text{otherwise.}
\end{cases}\label{eq:epsilon}
\end{equation}
A squarefree part of the denominator has degree in $y$ bounded by $d_x^\star+d_y^\star+\epsilon$.
{}From there, Theorem~\ref{th:Bronstein} yields $\bideg R\le(D_x,D_y)$, with
\begin{align}
D_x&:=2d_x^\star(d_x-d_x^\star+d_y-d_y^\star+1)+d_x(2(d_x^\star+d_y^\star+\epsilon)-1),\label{eq:DxDy}\\
\notag D_y&:=d_x^\star+d_y^\star+\epsilon.
\end{align}
\smallskip \noindent \textbf{Small branches.} It is classical that for a polynomial~$P=\sum{a_{i,j}x^iy^j}\in\KK[x,y]$, the number of its solutions tending to~0 can be read off its Newton polygon. This polygon is the lower convex hull of the union of $(i,j)+\mathbb{N}^2$ for $(i,j)$ such that~$a_{i,j}\neq0$. The number of solutions tending to~0 is given by the minimal $y$-coordinate of its leftmost points. Since the number of small branches counts only distinct solutions, it is thus given by
\begin{equation}\label{eq:Nsmall}
\nsmall(P)=\nsmall(P^\star)=\val_y([x^{\val_xP^\star}]P^\star).
\end{equation}

The change of variables~$x\mapsto x/y$ changes the coordinates of the point corresponding to~$a_{i,j}$ into~$(i,j-i)$. This transformation maps the vertices of the original Newton polygon to the vertices of the Newton polygon of the Laurent polynomial $P(x/y,y)$. Multiplying by~$y^{\slope(P)}$ yields a polynomial and shifts the Newton polygon up by~$\slope(P)$, thus
\[\nsmall\left(y^{\slope(P)}P(x/y,y)\right)=\nsmall(P^\star)+\slope(P^\star).\]

The number of small branches of the denominator of~$G$ constructed in the first step of the algorithm is then given by
\begin{equation}\label{eq:nb_small}
c:=\nsmall{(B^\star)}+\slope(B^\star)+\epsilon.
\end{equation}

\smallskip \noindent \textbf{Complexity.} We now analyze the cost of Algorithm~\ref{algo:diagonal}. The first step does not require any arithmetic operation. Next, the computation of~$R$ takes $\softO((d_x+d_y)^6)$ ops. (see the comment after Theorem~\ref{th:Bronstein}). The number of small branches is obtained with no arithmetic operation from a squarefree decomposition computed in Algorithm~\ref{algo:Bronstein}. Finally, Algorithm~\ref{algo:Sigma_c} uses~$\softO(cD_x\binom{D_y}{c}^2)$ ops.

We now have the values required by Theorem~\ref{thm:platypus-bound}, which concludes the proof of the following bounds.

\begin{thm} \label{thm:bound diagonals}
	Let $F=A/B$ be a rational function in $\KK(x,y)$ with $B(0,0)\neq0$. Let $(d_x, d_y)$ (resp. $(d_x^\star,d_y^\star)$) be a bound on the bidegrees of $A$ and~$B$ (resp. a squarefree part of~$B$). Let~$\epsilon,D_x,D_y,c$ be defined as in Eqs.~(\ref{eq:epsilon},\ref{eq:DxDy},\ref{eq:nb_small}).
	Then there exists a polynomial $\Phi\in\KK[t,\Delta]$ such that $\Phi(t, \Diag F(t))=0$ and 
  \[\bideg \Phi\le\left(D_x\binom{D_y}{c},\binom{D_y}{c}\right).\]
Algorithm~\ref{algo:diagonal} computes it in~$\softO \left(cD_x\binom{D_y}{c}^2+(d_x+d_y)^6 \right)$ ops.
\end{thm}
A general bound on $\bideg \Phi$ depending only on a bound $(d,d)$ on the bidegree of the input can be deduced from the above as
\[\bideg\Phi\le(d(4d+3),1)\times\binom{2d+1}{d}.\]

\subsection{Optimization}\label{subsec:Optimization}
Assume that the denominator of~$F(x/y)/y$ is already partially factored as $Q(y)=\tilde{Q}(y)\prod_{i=1}^k{(y-y_i(x))}$, where the $y_i$ are~$k$ distinct \emph{rational} branches among the~$c$ small branches of~$Q$. Then their corresponding (rational) residues $r_i$ contribute to the diagonal; therefore it is only necessary to invoke Algorithm~\ref{algo:diagonal} on $(\tilde{Q},c-k)$, which produces a polynomial $\tilde{\Phi}$. Then the polynomial $\Phi(t,\Delta) = \tilde{\Phi}(t, \Delta-\sum_{i}{r_i})$ cancels the diagonal of $F$.

In particular, this optimization applies systematically for the factor~$y^{-\alpha}$ when $\alpha<0$ (or equivalently $\epsilon=1$) in the algorithm. In this case, it yields a polynomial~$\Phi$ with smaller degree than the original algorithm:
\[\deg_\Delta\Phi\le\binom{d_x^\star+d_y^\star}{\nsmall(B^\star)+\slope(B^\star)}.\]
(A sharper bound on the degree in~$t$ can be derived as well.)

\subsection{Generic case}\label{subsec:Generic case}
The bounds from Theorem~\ref{thm:bound diagonals} on the bidegree of $\Phi$ are slightly pessimistic wrt the variable $t$, but generically tight wrt the variable~$\Delta$, as will be proved in Proposition~\ref{prop:generic} below. We first need a lemma.

\begin{lem}\label{lemma:galois groups}
	Let $\KK$ be a field of characteristic $0$, and $P\in \KK[y]$ be a polynomial of degree $d$, with Galois group $\mathfrak{S}_d$ over $\KK$. Assume that the roots $\alpha_1,\ldots\alpha_d$ of $P$ are algebraically independent over $\mathbb{Q}$.
	Then, for any $c\le d$, the degree $\binom{d}{c}$ polynomial $\Sigma_cP$ is irreducible in $\KK[y]$.
\end{lem}

\begin{proof}
	Since $\Sigma = \alpha_1 + \cdots+\alpha_c$ is a root of $\Sigma_c P$, it suffices to prove that $\KK(\Sigma)$ has degree $\binom{d}{c}$ over $\KK$.
	The $\alpha_i$'s being algebraically independent, any permutation $\sigma \in \mathfrak{S}_d$ of all the $\alpha_i$'s that leaves~$\Sigma$ unchanged has to preserve $\alpha_{c+1}+\cdots+\alpha_d$ as well. It follows that $\KK(\alpha_1,\ldots,\alpha_d)$ has degree~$c!(d-c)!$ over~$\KK(\Sigma)$ and degree~$d!$ over~$\KK$, so that~$\KK(\Sigma)$ has degree~$\binom{d}{c}$ over~$\KK$
\end{proof}

\begin{prop}\label{prop:generic}
	Let $A$ be a polynomial in $\QQ[x,y]_{d_x,d_y}$, and
	\[B(x,y)=\sum_{i\le d_x,j\le d_y}{b_{i,j}x^iy^j}\in\QQ[(b_{i,j});x,y],\]
	where the $b_{i,j}$ are indeterminates.
	Then the polynomial computed by Algorithm~\ref{algo:diagonal} with input $A/B$ is irreducible of degree $\binom{d_x+d_y}{d_x}$ over $\KK=\QQ((b_{i,j});x)$.
\end{prop}

\begin{proof}
	First apply the change of variables to obtain $G=P/Q$, with $Q(x,y) = \sum_{i, j}{b_{i,j}x^iy^{d_x-i+j}}$. Denote $d=d_x+d_y$. Then, the polynomial $Q(1,y)$ has the form $\sum_{j\le d}{t_jy^j}$ where the $t_j$'s are algebraically independent over $\mathbb{Q}$. Therefore, $Q(1,y)$ has Galois group $\mathfrak{S}_{d}$ over $\QQ(t_0,\ldots,t_d)$ and its roots are algebraically independent over $\mathbb{Q}$~\cite[\S 57]{Waerden49}. This property lifts to $Q(x,y)$ \cite[\S 61]{Waerden49}, which thus has Galois group $\mathfrak{S}_d$ and algebraically independent roots, denoted $y_1,\ldots,y_d$.
	
	Now define the polynomial $R(x,y) = \prod_{i}{(y-P(x,y_i)/\partial_yQ(x,y_i))}$. Since $Q$ has simple roots, this is exactly the polynomial that is computed by Algorithm~\ref{algo:Bronstein}. The family $\left\{P(x,y_i)/\partial_yQ(x,y_i)\right\}$ is algebraically independent, since any algebraic relation between them would induce one for the $y_i$'s by clearing out denominators. In particular, the natural morphism $\operatorname{Gal}(Q/\KK)=\mathfrak{S}_{d}\rightarrow \operatorname{Gal}(R/\KK)$ is injective, whence an isomorphism. (Here, $\operatorname{Gal}(P/\KK)$ denotes the Galois group of $P\in\KK[y]$ over~$\KK$.) Since an immediate investigation of the Newton polygon of $Q$ shows that it has $d_x$ small branches, we conclude using Lemma~\ref{lemma:galois groups}.
\end{proof}

Proposition~\ref{prop:generic} implies that
for a generic rational function~$A/B$ with $A\in\KK[x,y]_{d,d}$ and $B\in\KK[x,y]_{d+1,d+1}$, the degree of $\Phi$ in $\Delta$ is $\binom{2d}{d}$. This is indeed observed on random examples.

\begin{example} 
	We consider a rational function $F(x,y)=1/{B(x,y)}$, where $B(x,y)$ is a dense polynomial of bidegree $(d,d)$ chosen at random. For~$d=1,2,3,4$, algorithm \textbf{AlgebraicDiagonal}($F$) produces \emph{irreducible} outputs with bidegrees
$(2,2)$, $(16,6)$, $(108, 20)$, $(640, 70) $, that are matched by the formulas 
\begin{equation}\label{eq:bound_generic}
\left(2d^2\binom{2d-2}{d-1},\binom{2d}{d}\right),
\end{equation}
so that the bound on $\deg_\Delta\Phi$ is tight in this case and the irreducibility of the output shows that Theorem~\ref{thm:bound diagonals} cannot be improved further.
\end{example}

\smallskip

\section{Walks}\label{sec:walks}
The exponential degree of the minimal polynomial of a diagonal proved in Proposition~\ref{prop:generic} concerns more generally other sums of residues, since this is the step where the exponential growth of the algebraic equations appears. This includes in particular constant terms of rational functions in ${\mathbb C}(x)[[y]]$, that can also be written as contour integrals of rational functions around the origin.

By contrast, sums of residues of a rational function always satisfy a differential equation of only polynomial size~\cite{BoChChLi10}.
Thus, when an algebraic function appears to be connected to a sum of residues of a rational function, the use of this differential structure is much more adapted to the computation of series expansions, instead of going through a potentially large polynomial. 

As an example where this phenomenon occurs naturally, we consider here the enumeration of unidimensional lattice walks, following Banderier and
Flajolet~\cite{BanderierFlajolet2002} and
Bousquet-M\'elou~\cite{Bousquet2006}. Our goal in this section is to study,
from the algorithmic perspective, the series expansions of various generating
functions (for bridges, excursions, meanders) that have been identified as algebraic~\cite{BanderierFlajolet2002}. 
One of
our contributions is to point out that although algebraic series
can be expanded
fast~\cite{ChudnovskyChudnovsky1986,ChudnovskyChudnovsky1987a,BostanChyzakLecerfSalvySchost2007},
the pre-computation of a polynomial equation could have prohibitive cost.
We overcome this
difficulty by pre-computing differential (instead of polynomial) equations
that have polynomial size only, and using them to compute series expansions to
precision~$N$ for bridges, excursions and meanders in time quasi-linear
in~$N$.

\subsection{Preliminaries}

We start with some vocabulary on lattice walks.
A \emph{simple step} is a vector $(1, u)$ with $u\in\ZZ$. 
A \emph{step set} $S$ is a finite set of simple steps.
A \emph{unidimensional walk} in the plane $\mathbb{Z}^2$ built from  $S$ is a finite sequence $(A_0,A_1,\ldots,A_n)$ of points in $\mathbb{Z}^2$, such that $A_0 = (0,0)$ and $\overrightarrow{A_{k-1}A_k} = (1, u_k)$ with $(1, u_k)\in S$. In this case $n$ is called the \emph{length} of the walk, and $S$ is the \emph{step set} of the walk. The $y$-coordinate of the endpoint $A_n$, namely $\sum_{i=1}^{n}{u_i}$, is called the final altitude of the walk.
The characteristic polynomial of the step set $S$ is
$$\Gamma_S(y) = \sum_{(1,u)\in S}{y^u}.$$

Following Banderier and Flajolet, we consider three specific families of walks: bridges, excursions and meanders~\cite{BanderierFlajolet2002}. \emph{Bridges} are walks with final altitude~$0$, \emph{meanders} are walks confined to the upper half plane, and \emph{excursions} are bridges that are also meanders.

We define the full generating power series of walks \[W_S(x,y) = \sum_{n\ge 0,k\in\ZZ}{w_{n,k}x^ny^k} \; \in \mathbb{Z}[y,y^{-1}][[x]],\] where $w_{n,k}$ is the number of walks with step set~$S$, of length $n$ and final altitude $k$. We denote by $B_S(x)$ (resp. $E_S(x)$, and $M_S(x)$) the power series $\sum_{n\ge0}{u_{n}x^n}$, where $u_n$ is the number of bridges (resp. excursions, and meanders) of length~$n$ with step set~$S$.

We omit the step set~$S$ as a subscript when there is no ambiguity. Several properties of the power series $W$, $B$, $E$ and $M$ are classical:
\begin{fact}\cite[\S2.1-2.2]{BanderierFlajolet2002}\label{fact:walks}\ The power series $W$, $B$, $E$ and $M$ satisfy
	\begin{compactenum}
	\item[(1)] $W(x,y)$ is rational and $W(x,y) =1/(1-x\Gamma(y))$;
	\item[(2)] $B(x)$, $E(x)$ and $M(x)$ are algebraic;
	\item[(3)] $B(x)=[y^0]W(x,y)$;
	\item[(4)] $E(x) = \exp\left(\int{(B(x)-1)/x\,\mathrm{d}x}\right)$.
	\end{compactenum}
\end{fact}
Our main objective in what follows is to study the efficiency of computing the power series expansions of the series $B$, $E$ and $M$. In the next two sections, we first study two previously known methods, then we design a new one.

\subsection{Expanding the generating power series}
We denote by $u^-$ (resp. $u^+$) the largest $u$ such that $(1,-u)\in S$ (resp. $(1,u)\in S$) and denote by $d$ the sum $u^- + u^+$. The integer $d$ measures the vertical amplitude of $S$; this makes $d$ a good scale for measuring the complexity of the algorithms that will follow. We assume that both $u^-$ and $u^+$ are positive, since otherwise the study of the excursions and meanders becomes trivial. 

\smallskip \noindent \textbf{The direct method.}
The combinatorial definition of walks yields a recurrence relation for $w_{n,k}$:
\begin{equation}
\label{eq:rec}
 	w_{n, k} = \sum_{(1, u)\in S}{w_{n-1, k-u}},
\end{equation}
with initial conditions $w_{n,k} = 0$ if $n,k \le 0$ with $(n,k)\neq(0,0)$, and $w_{0,0} = 1$. If $\tilde{w}_{n,k}$ denotes the number of walks of length $n$ and final altitude $k$ that never exit the upper half plane, then $\tilde{w}_{n,k}$ also satisfies recurrence~\eqref{eq:rec}, but with the additional initial conditions $\tilde{w}_{n,k} = 0$ for all $k<0$. Then the bridges (resp. excursions, meanders) are counted by the numbers $w_{n,0}$ (resp. $\tilde{w}_{n,0}$, $\sum_k\tilde{w}_{n,k}$).

One can compute these numbers by unrolling the recurrence relation~\eqref{eq:rec}. Each use of the recurrence costs $\bigO(d)$ ops., and in the worst case one has to compute $\bigO(dN^2)$ terms of the sequence (for example, if the step set is $S = \{ (1,1), \ldots, (1,d)\} $). This leads to the computation of each of the generating series in $\bigO(d^2N^2)$ ops.

\smallskip \noindent \textbf{Using algebraic equations.}
Another method is suggested in \cite[\S 2.3]{BanderierFlajolet2002}. It relies on the algebraicity of $B$, $E$ and $M$ (Fact~\ref{fact:walks}{\em(2)}). 
The series $E$ and $M$ can be expressed as products in terms of the small branches of the characteristic polynomial $\Gamma_S$ (see \cite[Th. 1, Cor. 1]{BanderierFlajolet2002}). From there, a polynomial equation can be obtained using the Platypus algorithm~\cite[\S2.3]{BanderierFlajolet2002}, which computes a polynomial canceling the products of a fixed number of roots of a given polynomial.  Given a polynomial equation $P(z,E)=0$, another one for~$B$ can be deduced from the relation~$B=zE'/E+1$ as $\Resultant_E((B-1)EP_E+zP_z,P)$.

Once a polynomial equation is known for one of these three series, it can be used to compute a linear recurrence with polynomial coefficients satisfied by its coefficients~\cite{ChudnovskyChudnovsky1986,ChudnovskyChudnovsky1987a,BostanChyzakLecerfSalvySchost2007}. This method produces an algorithm that computes the first~$N$ terms of~$B$, $E$ and~$M$ in $\bigO(N)$ ops. For this to be an improvement over the naive method for large~$N$, the dependence on~$d$ of the constant in the $\bigO()$ should not be too large and the precomputation not too costly.

Indeed, the cost of the pre-computation of an algebraic equation is not negligible. Generically, the minimal polynomial of $E$ has degree $\binom{d}{u^-}$, which may be exponentially large with respect to $d$~\cite{Bousquet2006}. Empirically, the polynomials for $B$ and $M$ are similarly large. 

The situation for differential equations and recurrences is different: $B$ satisfies a differential equation of only polynomial size (see below), whereas (empirically), those for $E$ and $M$ have a potentially exponential size. These sizes then transfer to the corresponding recurrences and thereby to the constant in the complexity of unrolling them.
\begin{example} With the step set $S = \left\{(1,d), (1,1), (1,-d)\right\}$ and $d\ge 2$, the counting series $W_{S}$ equals
\[W_S(x,y) = \frac{y^d}{y^d-x(1+y^{d+1}+y^{2d})}.\]
Experiments indicate that the minimal polynomial of $B_S(x)$ has bidegree $(2d\binom{2d-2}{d-1}, \binom{2d}{d})$, exhibiting an exponential growth in~$d$. On the other hand, they show that $B_S(x)$ satisfies a linear differential equation of order $2d-1$ and coefficients of degree $d^2+3d-2$ for even~$d$, and $d^2+3d-4$ for odd $d$.
\end{example}

\smallskip \noindent {\bf New Method.}
We now give a method that runs in quasi-linear time (with respect to $N$) and avoids the computation of an algebraic equation. Our method relies on the fact that periods of rational functions such as the one in Part~{\em(3)} of Fact \ref{fact:walks} satisfy differential equations of polynomial size in the degree of the input rational function~\cite{BoChChLi10}.
We summarize our results in the following theorem, and then go over the proof in each case individually.

\begin{thm}\label{thm:walks} Let $S$ be a finite set of simple steps and $d=u^-+u^+$.
	The series $B_S$ (resp. $E_S$ and $M_S$) can be expanded at order $N$ in $\bigO(d^2N)$ ops. (resp. $\softO(d^2N)$ ops.), after a pre-computation in $\softO(d^5)$~ops.
\end{thm}
\subsection{Fast Algorithms}

\noindent \textbf{Bridges.} To expand $B(x)$, we rely on Fact~\ref{fact:walks}{\em(3)}. The formula can be written $B=(1/2\pi i)\oint{W(x,y)\frac{dy}{y}}$, the integration path being a circle inside a small annulus around the origin~\cite[proof of Th.~1]{BanderierFlajolet2002}. Moreover, $W(x,y)/y$ is of the form $P/Q$, where $\bideg Q \le (1, d)$ and $\bideg P \le (0, d-1)$. Since $P$ and $Q$ are relatively prime and $Q$ is primitive with respect to $y$, Algorithm \textbf{HermiteTelescoping}~\cite[Fig.~3]{BoChChLi10}  computes a telescoper for $P/Q$, which is also a differential equation satisfied by $B$, in  $\softO(d^5)$ ops. The resulting differential equation has order at most~$d$ and degree $O(d^2)$.
This differential equation can be turned into a recurrence of order $O(d^2)$ in quasi-optimal time (see the discussion after~\cite[Cor. 2]{BoSc05}). We may use it to expand $B(x) \bmod x^N$ in $O(d^2N)$ ops, once we have a way to compute the initial conditions. But this can be done using the naive algorithm described above in $\softO(d^4)$ ops. Thus, the total cost of the pre-computation is $\softO(d^5)$, as announced.

\begin{algo}
	Algorithm \textbf{Walks}($S$, $N$)
	
	\begin{algoenv}{A set $S$ of simple steps and an integer $N$}{$B_S,E_S,M_S\bmod x^{N+1}$}
		\State $F \gets W(x,y)/y$ [case $B, E$] or $W(x,y)/(1-y)$ [case $M$]
		\State $D\gets \textbf{HermiteTelescoping}(F)$ \cite[Fig.~3]{BoChChLi10}
		\State $R\gets $ the recurrence of order $r$ associated to $D$
		\State $I\gets [y^0]W(x,y)\bmod x^{r+1}$ [case $B,E$]\\ \qquad $[y^0]yW(x,y)/(1-y)\bmod x^{r+1}$ [case $M$]
		\State $B \gets  [y^0]W(x,y)\bmod x^{N+1}$ (from $R,I$)
		\State $A \gets  [y^0]yW(x,y)/(1-y)\bmod x^{N+1}$ (from $R,I$)
		\State $E\gets \exp\left(\int{(B(x)-1)/x\,\mathrm{d}x}\right)\bmod x^{N+1}$
		\State $M\gets \exp\left(-\int{(A(x)/x)/(1-\Gamma(1)x)}\,\mathrm{d}x\right)\bmod x^{N+1}$
		\State\Return$B,E,M$
	\end{algoenv}
	\caption{Expanding the generating functions\\ of bridges, excursions and meanders}\label{algo:bridges}
\end{algo}

\smallskip \noindent \textbf{Excursions.} If $B(x) \bmod x^{N+1}$ is known, it is then possible to recover $E(x) \bmod x^{N+1}$ thanks to Fact~\ref{fact:walks}{\em(4)}. Expanding $E(x)$ comes down to the computation of the exponential of a series, which can be performed using $\softO(N)$ ops. (Fact~\ref{fact:complexity}{\em(4)}).

\smallskip \noindent \textbf{Meanders.} As in the case of excursions, the logarithmic derivative of $M(x)$ is recovered from a sum of residues by the following.
\begin{prop} \label{walks-formulas2} The  series~$W$ and~$M$ are related through
\[A(x) = [y^0]\frac{y}{1-y}W(x,y),\quad M(x) = \frac{\exp\left(-\int\frac{A(x)}{x}\,\mathrm{d}x\right)}{1-x\Gamma(1)}.\]
\end{prop}
\begin{proof}
	Denote by $y_1,\ldots,y_{u^-}$ the small branches of the polynomial $y^{u^-} - xy^{u^-}\Gamma(y)$. Then $M$ is given as~\cite[Cor.~1]{BanderierFlajolet2002}:
  \[M(x)=\frac{1}{1-x\Gamma(1)}\prod_{i=1}^{u^{-}}{(1-y_i)}.\]
  On the other hand,
\begin{multline*}
    A(x) = \frac1{2\pi i} \oint{\frac{W(x,y)}{1-y}\mathrm{d}y}\\
    = \sum_{i=1}^{u^-}{\Residue_{y=y_i(x)}\left(\frac{1}{(1-y)(1-x\Gamma(y))}\right)} 
    =  -\sum_{i=1}^{u^-}{\frac{1}{(1-y_i)x\Gamma '(y_i)}},
  \end{multline*}
  where the integral has been taken over a circle around the origin and the small branches.
	Differentiating the equation $1-x\Gamma(y)$ = 0 with respect to $x$ leads to
$-x\Gamma '(y_i) = {1}/({xy_i'})$, whence
	$A(x) = x\sum_{i=1}^{u^{-}}{{y_i'}/({1-y_i})}.$ Therefore, $\prod ( 1- y_i) = \exp ( - \int A/x\,\mathrm{d}x) )$, finishing the proof.
\end{proof}
Thus we apply the same method as in the case of the excursions. We first compute a differential equation for $A(x)$ using the method of~\cite{BoChChLi10}. The computation of the initial conditions for $A$ can also be performed naively from its definition as a constant term, by simply expanding $yW(x,y)/(1-y)$. The formula of the proposition then recovers $M(x)$. The complexity analysis goes exactly as in the previous case, giving a global cost of $\softO(d^5)$ ops.

{\bf Acknowledgements.} This work has been supported in part by FastRelax ANR-14-CE25-0018-01.

\bibliographystyle{abbrv}

\begin{scriptsize}

\end{scriptsize}

\end{document}